\documentclass[conference]{IEEEtran}

\usepackage{graphicx}
\usepackage{amsmath}
\usepackage{mathrsfs}
\usepackage{amssymb}
\usepackage{float}
\usepackage{url}
\usepackage{verbatim}
\usepackage{dsfont}
\usepackage{enumerate}

\newtheorem{theorem}{Theorem}
\newtheorem{lemma}{Lemma}
\newtheorem{proposition}{Proposition}
\newtheorem{corollary}{Corollary}

\newtheorem{definition}{Definition}
\newtheorem{example}{Example}
\newtheorem{remark}{Remark}
\newcommand{\prob}{\ensuremath{\mathbb{P}}}

\newcommand{\Reals}{\ensuremath{\mathbb{R}}}
\newcommand{\set}{\ensuremath{\mathcal}}
\newcommand{\PU}{\ensuremath{P_1}}
\newcommand{\PZ}{\ensuremath{P_0}}

\DeclareMathOperator*{\esssup}{ess\,sup}
\DeclareMathOperator*{\essinf}{ess\,inf}

\newcommand{\MYfooter}{\smash{
\hfil\parbox[t][\height][t]{\textwidth}{}\hfil\hbox{}}}
\makeatletter
\def\ps@IEEEtitlepagestyle{%
\def\@oddhead{\mbox{}2016 ICSEE International Conference on the Science of Electrical Engineering \rightmark \hfil } \def\@oddfoot{\MYfooter{978-1-5090-2152-9/16/\$31.00\:\copyright2016\:IEEE}}%
\def\@evenfoot{\MYfooter}}
\def\ps@headings{%
\def\@oddhead{\mbox{}2016 ICSEE International Conference on the Science of Electrical Engineering \rightmark \hfil }} \makeatother 
\begin{document}
\title {\LARGE{$f$-Divergence Inequalities via Functional Domination}}
\author{
\IEEEauthorblockN{Igal Sason}
\IEEEauthorblockA{Andrew and Erna Viterbi Faculty of Electrical Engineering \\
Technion-Israel Institute of Technology \\
Haifa 32000, Israel\\
E-mail: sason@ee.technion.ac.il}
\and
\IEEEauthorblockN{Sergio Verd\'{u}}
\IEEEauthorblockA{Department of Electrical Engineering \\
Princeton University\\
New Jersey 08544, USA\\
E-mail: verdu@princeton.edu}}

\maketitle

\begin{abstract}
This paper considers derivation of $f$-divergence inequalities via the approach
of functional domination. Bounds on an $f$-divergence based on one or several
other $f$-divergences are introduced, dealing with pairs of probability measures
defined on arbitrary alphabets. In addition, a variety of bounds are shown to hold
under boundedness assumptions on the relative information.\footnote{This work has
been supported by the Israeli Science Foundation (ISF) under
Grant 12/12, by NSF Grant CCF-1016625, by the Center for Science of Information,
an NSF Science and Technology Center under Grant CCF-0939370, and by ARO under
MURI Grant W911NF-15-1-0479.}
\end{abstract}


\section{Basic Definitions}
\label{sec:preliminaries}

We assume throughout that the probability measures $P$ and $Q$ are
defined on a common measurable space $(\set{A}, \mathscr{F})$, and $P \ll Q$
denotes that $P$ is {\em absolutely continuous} with respect to $Q$.

\begin{definition} \label{def:RI}
If $P \ll Q$, the {\em relative information} provided by $a \in \set{A}$
according to $(P,Q)$ is given by\footnote{$\frac{\text{d}P}{\text{d}Q}$ denotes the
Radon-Nikodym derivative (or density) of $P$ with respect to $Q$. Logarithms
have an arbitrary common base, and the exponent indicates the inverse function
of the logarithm with that base.}
\begin{align}  \label{eq:RI}
\imath_{P\|Q}(a) \triangleq \log \frac{\text{d}P}{\text{d}Q} \, (a).
\end{align}
\end{definition}

Introduced by Ali-Silvey \cite{AliS} and Csisz\'ar (\cite{Csiszar63}),
a useful generalization of the relative entropy, which retains some of its
major properties (and, in particular, the data processing inequality),
is the class of $f$-divergences. A general
definition of $f$-divergence is given in \cite[p.~4398]{LieseV_IT2006},
specialized next to the case where $P \ll Q$.

\begin{definition} \label{def:fD}
Let $f \colon (0, \infty) \to \Reals$ be a convex function,
and suppose that $P \ll Q$. The {\em $f$-divergence} from $P$ to $Q$ is given
by
\begin{align} \label{eq:fD}
D_f(P\|Q) = \int f \left(\frac{\text{d}P}{\text{d}Q}\right) \, \text{d}Q
= \mathbb{E} \bigl[f(Z) \bigr]
\end{align}
with
\begin{align} \label{eq: Z}
Z = \exp\bigl(\imath_{P\|Q}(Y)\bigr), \quad Y \sim Q.
\end{align}
In \eqref{eq:fD}, we take the continuous extension\footnote{The
convexity of $f \colon (0, \infty) \to \Reals$ implies its continuity on $(0, \infty)$.}
\begin{align} \label{eq: f at 0}
f(0) = \lim_{t \downarrow 0} f(t) \in (-\infty, +\infty].
\end{align}

If $p$ and $q$ denote, respectively, the densities of $P$ and $Q$ with respect to a
$\sigma$-finite measure $\mu$ (i.e., $p = \frac{\text{d}P}{\text{d}\mu}$,
$q=\frac{\text{d}Q}{\text{d}\mu}$), then we can write \eqref{eq:fD} as
\begin{align} \label{eq:fD2}
D_f(P\|Q) = \int q \; f\left(\frac{p}{q}\right) \, \text{d}\mu.
\end{align}
\end{definition}

\begin{remark} \label{remark: equivalence-fD}
Different functions may lead to the same $f$-divergence for all $(P,Q)$: if for an
arbitrary $b \in \Reals$, we have
\begin{align}  \label{eq: fD3}
f_b(t) = f_0(t) + b \, (t-1), \quad t \geq 0
\end{align}
then
\begin{align}  \label{eq: fD4}
D_{f_0}(P\|Q) = D_{f_b}(P\|Q).
\end{align}
\end{remark}

\par
Relative entropy is $D_r(P\|Q)$ where
$r$ is given by
\begin{align} \label{eq: r}
r(t) = t \log t + (1-t) \log e,
\end{align}
and the total variation distance $|P-Q|$ and $\chi^2$ divergence
$\chi^2(P\|Q)$ are $f$-divergences with $f(t) = (t-1)^2$ and $f(t)=|t-1|$,
respectively.

\par
The following key property of $f$-divergences follows from Jensen's inequality.

\begin{proposition} \label{prop: fGibbs}
If $f \colon (0, \infty) \to \Reals$ is convex and $f(1)=0$,
$P \ll Q$, then
\begin{align}  \label{eq: fGibbs}
D_f(P\|Q) \geq 0.
\end{align}
If, furthermore, $f$ is strictly convex at $t=1$, then equality in \eqref{eq: fGibbs}
holds if and only if $P=Q$.
\end{proposition}

The reader is referred to \cite{Vajda_2009} for a survey on general properties
of $f$-divergences, and also to the textbook by Liese and Vajda \cite{LieseV_book87}.

The numerical optimization of an $f$-divergence subject to simultaneous
constraints on $f_i$-divergences $(i=1, \ldots , L)$
was recently studied in \cite{GSS_IT14}, which showed that for that purpose
it is enough to restrict attention to alphabets of cardinality $L+2$.

The full paper version of our work, which includes several approaches
for the derivation of $f$-divergence inequalities, is available in \cite{SV15}.

\section{Functional Domination} \label{sec: functional domination}
Let $f$ and $g$ be convex functions on $(0, \infty)$ with $f(1)=g(1)=0$, and let $P$ and $Q$
be probability measures defined on a measurable space $(\set{A}, \mathscr{F})$.
If, for $\alpha > 0$, $f(t) \leq \alpha g(t)$ for all $t \in (0, \infty)$ then,
it follows from Definition~\ref{def:fD} that
\begin{align} \label{000}
D_f(P \|Q) \leq \alpha\, D_g(P \| Q).
\end{align}
This simple observation leads to a proof of several inequalities
with the aid of Remark~\ref{remark: equivalence-fD}.

\subsection{Basic Tool} \label{subsec: basic tool of functional domination}

We start this section by proving a general result, which will be helpful in proving various
tight bounds among $f$-divergences.
\begin{theorem} \label{theorem: tight bound}
Let $P \ll Q$, and assume
\begin{itemize}
\item $f$ is convex  on
$(0, \infty)$ with $f(1)=0$;
\item
$g$ is convex  on
$(0, \infty)$ with $g(1)=0$;
\item
$g(t) > 0$ for all $t \in (0,1) \cup (1, \infty)$.
\end{itemize}
 Denote the function $\kappa\colon (0,1) \cup (1, \infty) \to \Reals$
\begin{align} \label{kappadef-1}
\kappa(t) &= \frac{f(t)}{g(t)}, \quad t \in (0,1) \cup (1, \infty)
\end{align}
and
\begin{align}
\label{barkdef}
\bar{\kappa} &= \sup_{t  \in (0,1) \cup (1, \infty)} \kappa(t).
\end{align}
Then,
\begin{enumerate}[a)]
\item \label{theorem: tight bound: parta}
\begin{align}  \label{eq: tight bound--0}
D_f(P \| Q) \leq  \bar{\kappa} \, D_g(P \| Q).
\end{align}
\item \label{theorem: tight bound: partb}
If, in addition, $f'(1)=g'(1)=0$, then
\begin{align}  \label{eq: tight bound}
\sup_{P \neq Q} \frac{D_f(P \| Q)}{D_g(P \| Q)} = \bar{\kappa}.
\end{align}
\end{enumerate}
\end{theorem}

\medskip

\begin{proof}
See \cite[Theorem~1]{SV15}.
\end{proof}

\begin{remark}
Beyond the restrictions in Theorem~\ref{theorem: tight bound}\ref{theorem: tight bound: parta}),
the only operative restriction imposed by
Theorem~\ref{theorem: tight bound}\ref{theorem: tight bound: partb}) is the differentiability
of the functions $f$ and $g$ at $t=1$. Indeed, we can invoke Remark~\ref{remark: equivalence-fD}
and add $f'(1) \, (1-t)$ to $f(t)$, without changing $D_f$ (and likewise with $g$) and thereby
satisfying the condition in Theorem~\ref{theorem: tight bound}\ref{theorem: tight bound: partb});
the stationary point at~1 must be a minimum of both $f$ and $g$ because of the assumed
convexity, which implies their non-negativity on $(0, \infty)$.
\end{remark}

\vspace*{0.1cm}
\begin{remark}
It is useful to generalize Theorem~\ref{theorem: tight bound}\ref{theorem: tight bound: partb})
by dropping the assumption on the existence of the derivatives at~1. As it is explained
in \cite{SV15}, it is enough to require that the left derivatives of $f$ and $g$ at~1 be
equal to $0$. Analogously, if $\bar{\kappa} = \sup_{0<t<1} \kappa(t)$, it is enough to
require that the right derivatives of $f$ and $g$ at~1 be equal to $0$.
\end{remark}

\subsection{Relationships Among $D(P \| Q)$, $\chi^2(P \| Q)$ and $|P-Q|$}
\label{subsec: RE-chi2-TV functional domination}

\begin{theorem} \
\begin{enumerate}[a)]
\item
If $P \ll Q$ and $c_1, c_2 \geq 0$, then
\begin{align} \label{eq: Improved DiaconisS96}
D(P \| Q) \leq \left(c_1 \, |P-Q| + c_2 \, \chi^2(P \| Q) \right) \log e
\end{align}
holds if $(c_1, c_2) = (0,1)$ and $(c_1, c_2) = \bigl(\tfrac14, \tfrac12\bigr)$. Furthermore,
if $c_1=0$ then $c_2=1$ is optimal, and if $c_2 = \tfrac12$ then $c_1 = \tfrac14$ is optimal.
\item
If $P \ll \gg Q$ and $P \neq Q$, then
\begin{align}
\label{eq: symmetrized RE-chi^2}
\frac{D(P \| Q) + D(Q \| P)}{\chi^2(P \| Q) + \chi^2(Q \| P)} \leq \tfrac12 \, \log e
\end{align}
and the constant in the right side of \eqref{eq: symmetrized RE-chi^2} is the best possible.
\end{enumerate}
\end{theorem}

\begin{proof}
See \cite[Theorem~2]{SV15}.
\end{proof}

\begin{remark}
Inequality~\eqref{eq: Improved DiaconisS96} strengthens the bound in \cite[(2.8)]{DiaconisS96},
\begin{align} \label{eq: DiaconisS96}
D(P \| Q) \leq \tfrac12 \left(|P-Q| + \chi^2(P \| Q)\right) \log e.
\end{align}
Note that the short outline of the suggested proof in \cite[p.~710]{DiaconisS96} leads not
\eqref{eq: DiaconisS96} but to the weaker upper bound
$|P-Q| + \tfrac12 \, \chi^2(P \| Q)$ nats.
\end{remark}

\subsection{An Alternative Proof of Samson's Inequality}
\label{subsec: Samson's Inequality}

For the purpose of this sub-section, we introduce {\em Marton's divergence} \cite{Marton96}:
\begin{align} \label{smarton}
d_2^2(P,Q)  = \min \mathbb{E} \left[ \prob^2[ X \neq Y \, | \, Y ] \right]
\end{align}
where the minimum is over all probability measures $P_{XY}$ with respective
marginals $P_X=P$ and $P_Y=Q$.
From \cite[pp.~558--559]{Marton96}
\begin{align} \label{smarton as f-div}
d_2^2(P,Q) = D_s( P\| Q)
\end{align}
with
\begin{align} \label{eq: s}
s(t) = (t - 1)^2 \; 1\{ t < 1 \}.
\end{align}
Note that Marton's divergence satisfies the triangle inequality \cite[Lemma~3.1]{Marton96},
and $d_2(P,Q)=0$ implies $P=Q$; however, due to its asymmetry, it is not a distance measure.

An analog of Pinsker's inequality, which comes in handy for the proof of Marton's conditional
transportation inequality \cite[Lemma~8.4]{boucheron2013concentration}, is the following bound
due to Samson \cite[Lemma~2]{samson2000concentration}:
\begin{theorem} \label{thm: Samson}
If $P \ll Q$, then
\begin{align} \label{eq: Samson}
d_2^2(P,Q) + d_2^2(Q,P) \leq \tfrac2{\log e} \; D(P\|Q).
\end{align}
\end{theorem}

In \cite[Section~3.D]{SV15}, we provide an alternative proof of Theorem~\ref{thm: Samson}, in view of
Theorem~\ref{theorem: tight bound}\ref{theorem: tight bound: partb}), with the following advantages:
\begin{enumerate}[a)]
\item This proof yields the optimality of the constant in \eqref{eq: Samson}, i.e., we prove that
\begin{align}  \label{eq: Samson - tight}
\sup_{P \neq Q} \frac{d_2^2(P,Q) + d_2^2(Q,P)}{D(P \| Q)} = \tfrac2{\log e}
\end{align}
where the supremum is over all probability measures $P,Q$ such that $P \neq Q$ and $P \ll \gg Q$.
\item A simple adaptation of this proof results in a reverse inequality to \eqref{eq: Samson}, which holds
under the boundedness assumption of the relative information (see Section~\ref{subsec: Reverse Samson's Inequality}).
\end{enumerate}

\subsection{Ratio of $f$-Divergence to Total Variation Distance}
\label{subsec: f-divergence and total variation distance}
Let $f \colon (0, \infty) \to \Reals$ be a convex function with $f(1)=0$, and let
$f^\star \colon (0, \infty) \to \Reals$ be given by
\begin{align} \label{eq: fstar}
f^\star(t) = t \, f\left(\tfrac1t\right)
\end{align}
for all $t > 0$. Note that $f^\star$ is also convex, $f^\star(1)=0$, and
$D_f(P \| Q) = D_{f^\star}(Q \| P)$ if $P \ll \gg Q$. By definition, we take
\begin{align}
\label{eq: fstar at 0}
f^\star(0) = \lim_{t \downarrow 0} f^\star(t) = \lim_{u \to \infty} \frac{f(u)}{u}.
\end{align}

Vajda \cite[Theorem~2]{Vajda_1972} showed that the range of an $f$-divergence
is given by
\begin{align} \label{eq: Vajda72}
0 \leq D_f(P \| Q) \leq f(0) + f^{\star}(0)
\end{align}
where every value in this range is attainable by a suitable pair of probability
measures $P \ll Q$. Recalling  Remark \ref{remark: equivalence-fD}, note that
$f_b(0) + f_b^{\star}(0) = f(0) + f^{\star}(0)$ with $f_b(\cdot)$ defined in \eqref{eq: fD3}.
Basu \textit{et al.} \cite[Lemma~11.1]{BasuSP} strengthened \eqref{eq: Vajda72},
showing that
\begin{align}
\label{eq: BasuSP inequality}
D_f(P \| Q) \leq \tfrac12 \left( f(0) +  f^\star(0)\right) \, |P-Q|.
\end{align}
If $f(0)$ and $f^\star(0)$ are finite,  \eqref{eq: BasuSP inequality}
yields a counterpart to a result by Csisz\'{a}r (see
\cite[Theorem~3.1]{Csiszar72}) which implies that if $f \colon (0, \infty) \to \Reals$ is
a strictly convex function, then there exists a real-valued function $\psi_f$
such that $\lim_{x \downarrow 0} \psi_f(x) = 0$, and
\begin{align} \label{eq: Csiszar72}
|P-Q| \leq \psi_f \bigl( D_f(P \| Q) \bigr).
\end{align}
Next, we demonstrate that the constant in \eqref{eq: BasuSP inequality} cannot be improved.

\begin{theorem} \label{theorem: strengthened BasuSP}
If $f \colon (0, \infty) \to \Reals$ is convex with $f(1)=0$, then
\begin{align}
\label{eq: BasuSP tightness}
\sup_{P \neq Q} \frac{D_f(P \| Q)}{|P-Q|} &= \tfrac12 \left( f(0) +  f^\star(0) \right)
\end{align}
where the supremum is over all probability measures $P, Q$ such that $P \ll Q$ and $P \neq Q$.
\end{theorem}

\begin{proof}
See \cite[Theorem~5]{SV15}.
\end{proof}

\begin{remark}
Csisz\'ar \cite[Theorem~2]{Csiszar67b} showed that if $f(0)$ and $f^\star(0)$ are finite
and $P \ll Q$, then there exists a constant $C_f > 0$ which depends only on $f$ such that
$D_f(P \| Q) \leq C_f \, \sqrt{|P-Q|}$. Note that,
if $|P-Q| < 1$, then this inequality is superseded by \eqref{eq: BasuSP inequality}
where the constant is not only explicit but is the best possible according to
Theorem~\ref{theorem: strengthened BasuSP}.
\end{remark}

A direct application of Theorem~\ref{theorem: strengthened BasuSP} yields
\begin{corollary}
\begin{align}
\label{eq: sup4}
& \sup_{P \neq Q} \frac{d_2^2(P,Q)}{|P-Q|} = \frac12, \\[0.1cm]
\label{eq: sup5}
& \sup_{P \neq Q} \frac{d_2^2(P,Q) + d_2^2(Q,P)}{|P-Q|} = 1
\end{align}
where the supremum in \eqref{eq: sup4} is over all $P \ll Q$ with $P \neq Q$, and
the supremum in \eqref{eq: sup5} is over all $P \ll \gg Q$ with $P \neq Q$.
\end{corollary}
\begin{proof}
See \cite[Corollary~1]{SV15}.
\end{proof}

\begin{remark}
The results in \eqref{eq: sup4} and \eqref{eq: sup5} form counterparts of
\eqref{eq: Samson - tight}.
\end{remark}

\section{Bounded Relative Information} \label{sec:bounded}
In this section we show that it is possible to find bounds among $f$-divergences
without requiring a strong condition of functional domination (see Section~\ref{sec: functional domination})
as long as the relative information is upper and/or lower bounded almost surely.

\subsection{Definition of $\beta_1$ and $ \beta_2$.}
The following notation is used throughout the rest of the paper. Given a pair of probability measures $(P,Q)$
on the same measurable space, denote $\beta_1, \beta_2 \in [0,1]$ by
\begin{align}
\label{eq: beta1}
\beta_1
&= \exp\bigl(- D_\infty (P\|Q)\bigr), \\[0.1cm]
\label{eq: beta2}
\beta_2
&= \exp\bigl(-D_\infty (Q\|P)\bigr)
\end{align}
with the convention that if $D_\infty (P\|Q) = \infty$, then $\beta_1 = 0$,
and if $D_\infty (Q\|P) = \infty$, then $\beta_2 = 0$. Note that if $\beta_1 >0$, then $P \ll Q$,
while $\beta_2 >0$ implies $Q \ll P$. Furthermore, if $P \ll \gg Q$, then with $Y \sim Q$,
\begin{align}
\label{eq: beta1-alt}
\beta_1 &= \essinf \frac{\text{d}Q}{\text{d}P} \, (Y) = \left( \esssup \frac{\text{d}P}{\text{d}Q} \, (Y) \right)^{-1},\\
\label{eq: beta2-alt}
\beta_2 &= \essinf \frac{\text{d}P}{\text{d}Q} \, (Y) = \left( \esssup \frac{\text{d}Q}{\text{d}P} \, (Y) \right)^{-1}.
\end{align}
The following example illustrates an important case in which $\beta_1$ and
$\beta_2$ are positive.

\begin{example} (\textit{Shifted Laplace distributions.}) \label{example: two Laplacians}
Let $P$ and $Q$ be the probability measures whose probability density functions
are, respectively, given by $f_{\lambda} ( \cdot - a_0 )$
and $f_{\lambda}( \cdot - a_1 )$ with
\begin{align} \label{eq: two Laplacians}
f_{\lambda} (x) =  \tfrac{\lambda}{2} \, \exp(-\lambda |x|), \quad x \in \Reals
\end{align}
where $\lambda > 0$. In this case,
\eqref{eq: two Laplacians} yields
\begin{align}  \label{eq: beta - 2 Laplacians}
\beta_1 = \beta_2 = \exp\bigl(-\lambda \, |a_1-a_0| \bigr) \in (0,1].
\end{align}
\end{example}



\subsection{Basic Tool} \label{subsec:bounds among fD}
Since $\beta_1 =1 \Leftrightarrow \beta_2 =1 \Leftrightarrow P = Q$, it
is advisable to avoid trivialities by excluding that case.
\begin{theorem} \label{thm: fD1}
Let $f$ and $g$ satisfy the assumptions in Theorem~\ref{theorem: tight bound},
and assume that $(\beta_1, \beta_2) \in [0,1)^2$. Then,
\begin{align} \label{eq:fD bound2}
D_f(P \| Q)
&\leq \kappa^\star \; D_g(P \| Q)
\end{align}
where
\begin{align}
\kappa^\star = \sup_{\beta \in (\beta_2, 1) \cup (1,\beta_1^{-1})} \kappa(\beta)
\end{align}
and $\kappa(\cdot)$ is defined in \eqref{kappadef-1}.
\end{theorem}

\begin{proof}
See \cite[Theorem~5]{SV15}.
\end{proof}

Note that if $\beta_1 = \beta_2 = 0$, then Theorem~\ref{thm: fD1} does not improve
upon Theorem~\ref{theorem: tight bound}\ref{theorem: tight bound: parta}).

\vspace*{0.1cm}
\begin{remark} \label{remark: play}
In the application of Theorem~\ref{thm: fD1}, it is often convenient to make use of the
freedom afforded by Remark~\ref{remark: equivalence-fD} and choose the corresponding offsets
such that:
\begin{itemize}
\item the positivity property of $g$ required by Theorem~\ref{thm: fD1} is satisfied;
\item the lowest $\kappa^\star$ is obtained.
\end{itemize}
\end{remark}

\begin{remark} \label{remark: tight constants}
Similarly to the proof of Theorem~\ref{theorem: tight bound}\ref{theorem: tight bound: partb}),
under the conditions therein, one can verify that the constants in Theorem~\ref{thm: fD1}
are the best possible among all probability measures $P,Q$ with given $(\beta_1, \beta_2) \in [0,1)^2$.
\end{remark}

\begin{remark}\label{remark:reverse}
Note that if we swap the assumptions on $f$ and $g$ in Theorem~\ref{thm: fD1}, the same result translates into
\begin{align}\label{eq:fD bound1}
\inf_{\beta \in (\beta_2, 1) \cup (1,\beta_1^{-1})} \kappa(\beta) \cdot D_g(P \| Q) \leq D_f(P \| Q).
\end{align}
Furthermore, provided both $f$ and $g$ are positive (except at $t=1$) and $\kappa$ is monotonically increasing,
Theorem~\ref{thm: fD1} and \eqref{eq:fD bound1} result in
\begin{align}
\kappa(\beta_2) \, D_g(P \| Q) & \leq D_f(P \| Q) \label{eq:fD bound1_pc} \\
&\leq \kappa(\beta_1^{-1}) \, D_g(P \| Q). \label{eq:fD bound2_pc}
\end{align}
In this case, if $\beta_1 > 0$, sometimes it is convenient to replace $\beta_1 > 0$ with
$\beta_1^\prime \in (0, \beta_1)$
at the expense of loosening the bound. A similar observation applies to $\beta_2$.
\end{remark}

\begin{example} \label{example: non-monotonic}
If $f(t) = (t-1)^2$ and $g(t)= |t-1|$, we get
\begin{align} \label{eq1: chi square - TV}
\chi^2 (P\|Q) \leq \max\{ \beta_1^{-1}-1, 1 - \beta_2\}  \; |P-Q|.
\end{align}
\end{example}

\subsection{Bounds on $\frac{D(P\|Q)}{D(Q\|P)}$}
\label{subsec: Bounds on RE/dual}

The remaining part of this section is devoted to various
applications of Theorem~\ref{thm: fD1}. From this point,
we make use of the definition of
$r \colon (0, \infty) \to [0, \infty)$ in \eqref{eq: r}.

An illustrative application of Theorem~\ref{thm: fD1} gives
upper and lower bounds on the ratio of relative entropies.

\begin{theorem}  \label{thm: bounds RE and dual}
Let $P \ll \gg Q$, $P \neq Q$, and $(\beta_1, \beta_2) \in (0,1)^2$.
Let $\kappa \colon (0,1)\cup(1,\infty) \to (0, \infty)$ be defined as
\begin{align}
\label{eq: kappa RE and dual}
\kappa(t) = \frac{t \log t + (1-t) \, \log e}{(t-1) \log e - \log t}.
\end{align}
Then,
\begin{align}
\label{eq: bounds RE and dual}
\kappa(\beta_2)  \leq \frac{D(P \| Q)}{D(Q \| P)} \leq \kappa(\beta_1^{-1}).
\end{align}
\end{theorem}

\begin{proof}
See \cite[Theorem~6]{SV15}.
\end{proof}

\subsection{Reverse Samson's Inequality}
\label{subsec: Reverse Samson's Inequality}
The next result gives a counterpart to Samson's inequality \eqref{eq: Samson}.

\begin{theorem} \label{thm: Samson - refined}
Let  $(\beta_1, \beta_2) \in (0,1)^2$.
Then,
\begin{align}
\label{eq: Samson - refined}
\inf \frac{d_2^2(P,Q) + d_2^2(Q,P)}{ D(P\|Q)} = \min \bigl\{\kappa(\beta_1^{-1}), \, \kappa(\beta_2) \bigr\}
\end{align}
where the infimum is over all $P \ll Q$ with given $(\beta_1, \beta_2)$, and
where $\kappa\colon (0,1)\cup(1,\infty) \to \bigl(0, \tfrac{2}{\log e} \bigr)$
is given by
\begin{align}
\kappa(t) = \frac{(t-1)^2}{r(t) \, \max\{1,t\}} , \quad t \in (0,1) \cup (1, \infty).
\label{eq: kappa1 - Samson}
\end{align}
\end{theorem}
\begin{proof}
See \cite[Theorem~7]{SV15}.
\end{proof}

\subsection{Local Behavior of $f$-Divergences}
\label{sec:On the Local Behavior of f-divergences}
Another application of Theorem~\ref{thm: fD1} shows that the local behavior of $f$-divergences
differs by only a constant, provided that the first distribution approaches the reference measure
in a certain strong sense.

\begin{theorem} \label{thm: local behavior fD}
Suppose that $\{P_n\}$, a sequence of probability measures defined on a measurable space
$(\set{A}, \mathscr{F})$, converges to $Q$ (another probability measure on the same space)
in the sense that, for $Y \sim Q$,
\begin{align}
\label{eq: 1st condition}
\lim_{n \to \infty} \esssup \frac{\text{d}P_n}{\text{d}Q} \, (Y) = 1
\end{align}
where it is assumed that $P_n \ll Q$ for all sufficiently large $n$.
If $f$ and $g$ are convex on $(0, \infty)$ and they are positive except at $t=1$
(where they are 0), then
\begin{align}  \label{eq: limit fD, Pn-->Q}
\lim_{n \to \infty} D_f(P_n \| Q) = \lim_{n \to \infty} D_g(P_n \| Q) = 0,
\end{align}
and
{\small
\begin{align}
\min\{ \kappa ( 1^-) , \kappa (1^+) \} \leq \lim_{n \to \infty}
\frac{D_f(P_n \| Q)}{D_g(P_n \| Q)} \leq \max\{ \kappa ( 1^-) , \kappa (1^+) \}
\label{eq: limit of ratio of f-divergences}
\end{align}}
\hspace*{-0.17cm} where we have indicated the left and right limits of the function $\kappa (\cdot)$,
defined in \eqref{kappadef-1}, at $1$ by $\kappa ( 1^-)$ and $\kappa (1^+)$, respectively.
\end{theorem}

\begin{proof}
See \cite[Theorem~9]{SV15}.
\end{proof}

\begin{corollary}\label{cor:ratios}
Let $\{P_n \ll Q \}$ converge to $Q$ in the sense of  \eqref{eq: 1st condition}.
Then,  $D(P_n \| Q)$ and $D(Q \| P_n)$ vanish
as $n \to \infty$ with
\begin{align} \label{eq: limit of RE/dual}
\lim_{n \to \infty} \frac{D(P_n \| Q)}{D(Q \| P_n)} &= 1.
\end{align}
\end{corollary}

\begin{corollary}\label{cor:ratios:chi}
Let $\{P_n \ll Q \}$ converge to $Q$ in the sense of  \eqref{eq: 1st condition}.
Then, $\chi^2(P_n \| Q)$ and $D(P_n \| Q)$ vanish as $n \to \infty$ with
\begin{align}\label{eq:cor:ratios:chi}
\lim_{n \to \infty} \frac{D(P_n \| Q)}{\chi^2(P_n \| Q)} &= \tfrac12 \log e.
\end{align}
\end{corollary}
Note that \eqref{eq:cor:ratios:chi} is known in the  finite alphabet case \cite[Theorem~4.1]{CsiszarS_FnT}).

\subsection{Strengthened Jensen's inequality} \label{subsec: superjensen}

Bounding away from zero a certain density between two probability measures
enables  the following
strengthened version of Jensen's inequality, which generalizes a result
in \cite[Theorem~1]{Dragomir06}.

\begin{lemma}\label{lemma: superjensen}
Let $f \colon \Reals \to \Reals$ be a convex function, $\PU \ll \PZ$ be probability measures defined on a measurable space $(\set{A}, \mathscr{F})$, and
fix an arbitrary random transformation $P_{Z|X}\colon \set{A} \to \Reals$.
Denote\footnote{We follow the notation in \cite{Verdu_book} where $\PZ \to P_{Z|X} \to P_{Z_\mathtt{0}}$
means that the marginal probability measures of the joint distribution $\PZ  P_{Z|X}$ are
$\PZ$ and $ P_{Z_\mathtt{0}}$.}  $\PZ \to P_{Z|X} \to P_{Z_\mathtt{0}}$,
and $\PU \to P_{Z|X} \to P_{Z_\mathtt{1}}$.
Then,
\begin{align}
& \beta \, \bigl(\mathbb{E} \left[ f ( \mathbb{E} [ Z_\mathtt{0} | X_\mathtt{0} ] ) \right]
- f ( \mathbb{E} [ Z_\mathtt{0} ]  ) \bigr)  \nonumber \\
\label{eq: lemma-superjensen}
& \leq \mathbb{E} [ f ( \mathbb{E} [ Z_\mathtt{1} | X_\mathtt{1} ] ) ]
- f ( \mathbb{E} [ Z_\mathtt{1} ]  )
\end{align}
where $X_0 \sim \PZ$, $X_1 \sim \PU$, and
\begin{align}  \label{eq: beta}
\beta \triangleq \essinf \frac{\mathrm{d}\PU}{\mathrm{d}\PZ} \, (X_0).
\end{align}
\end{lemma}

\begin{proof}
See \cite[Lemma~1]{SV15}.
\end{proof}

\begin{remark}
Letting $Z = X$, and choosing $\PZ$ so that $\beta = 0$ (e.g., $\PU$ is a restriction of $\PZ$
to an event of $\PZ$-probability less than~1), \eqref{eq: lemma-superjensen} becomes Jensen's
inequality $f(\mathbb{E}[X_\mathtt{1}]) \leq \mathbb{E}[f(X_\mathtt{1})]$.
\end{remark}

Lemma~\ref{lemma: superjensen} finds the following application to the derivation of $f$-divergence inequalities.

\begin{theorem} \label{thm: GI fD}
Let $f \colon (0, \infty) \to \Reals$ be a convex function with $f(1)=0$.
Fix $P \ll Q$ on the same space with $(\beta_1, \beta_2) \in [0,1)^2$ and let $X \sim P$. Then,
\begin{align}
\beta_2 \, D_f(P\|Q) &\leq \mathbb{E} \left[ f \left( \exp ( \imath_{P\|Q} (X) )
\right) \right]  - f\bigl( 1 + \chi^2(P\|Q) \bigr) \nonumber \\
\label{eq: GI fD}
&\leq \beta_1^{-1} \, D_f(P \| Q).
\end{align}
\end{theorem}

Specializing Theorem~\ref{thm: GI fD} to the convex function on $(0, \infty)$ where
$f(t) = -\log t$ sharpens the inequality
\begin{align} \label{grout425}
D ( P \| Q) &\leq \log \left( 1 + \chi^2 ( P \| Q ) \right) \\
\label{eq: CsiszarT06}
&\leq \chi^2(P \| Q) \log e.
\end{align}
under the assumption of bounded relative information.

\begin{theorem} \label{thm:d-chi}
Fix $P \ll \gg Q$ such that $(\beta_1, \beta_2) \in (0,1)^2$. Then,
\begin{align}
\label{lbchi}
\beta_2 \, D(Q\|P) &\leq \log \bigl(1 + \chi^2(P\|Q) \bigr) - D(P\|Q) \\
\label{ubchi}
& \leq \beta_1^{-1} \, D(Q\|P).
\end{align}
\end{theorem}

\section{Reverse Pinsker Inequalities}
\label{sec:reverseP}
It is not possible to lower bound $|P - Q|$ solely in terms of $D(P\|Q)$ since for an
arbitrary small $\epsilon >0$ and an arbitrary large $\lambda >0$, we can construct
examples with $|P - Q| < \epsilon$ and $\lambda < D(P\|Q) < \infty$.
As in Section \ref{sec:bounded},  the following result involves the bounds on the relative information.
\begin{theorem} \label{thm: improved SV-ITA14}
If $\beta_1 \in (0, 1)$ and $\beta_2 \in [0, 1)$, then,
\begin{align} \label{eq: improved SV-ITA14}
D(P\|Q) \leq \tfrac12 \left( \varphi(\beta_1^{-1}) - \varphi(\beta_2) \right) \, |P-Q|
\end{align}
where $\varphi \colon [0, \infty) \to [0, \infty)$ is given by
\begin{align}
\varphi (t) =
\left\{
\begin{array}{ll}
0 & t=0\\
\frac{t \log t}{t-1} & t \in (0,1) \cup (1,\infty) \\
\log e & t=1.
\end{array}
\right.
\end{align}
\end{theorem}

\begin{proof}
See \cite[Theorem~23]{SV15}.
\end{proof}

\begin{remark}
Note that for Theorem~\ref{thm: improved SV-ITA14} to give a nontrivial result,
it is necessary that the relative information be upper bounded, namely $\beta_1 >0$.
However, we still get a nontrivial bound if $\beta_2 =0$.
\end{remark}

In the following, we assume that $P$ and $Q$ are probability measures defined
on a common finite set $\set{A}$, and $Q$ is strictly positive on $\set{A}$
with $|\set{A}| \geq 2$.
\begin{theorem} \label{thm: UB-RE-FS}
Let $Q_{min} = \min_{a \in \set{A}} Q(a)$, then
\begin{align}
\label{eq: UB-RE-FS2}
D(P \| Q) \leq \log \left(1 + \frac{|P-Q|^2}{2 Q_{\min}} \right).
\end{align}
Furthermore, if $Q \ll P$ and $\beta_2$ is defined as in \eqref{eq: beta2},
then the following tightened bound holds:
\begin{align*}
D(P \| Q) \leq \log \left(1 + \frac{|P-Q|^2}{2 Q_{\min}} \right)
- \tfrac12 \beta_2 |P-Q|^2 \log e.
\end{align*}
\end{theorem}
\begin{proof}
See \cite[Theorem~25]{SV15}.
\end{proof}

\begin{remark}
The result in Theorem~\ref{thm: UB-RE-FS} improves the inequality by Csisz\'ar and Talata
\cite[p.~1012]{CsiszarT_IT06}:
\begin{align}  \label{eq: CsTa}
D(P \| Q) \leq \left(\frac{\log e}{Q_{\min}} \right) \cdot |P-Q|^2.
\end{align}
\end{remark}

For further reverse Pinsker Inequalities and some of their implications,
see \cite[Section~6]{SV15}.

\end{document}